\documentclass[11pt,a4paper]{article}
\usepackage{amsfonts}
\usepackage[utf8]{inputenc}
\usepackage[english]{babel}
\usepackage{paralist}
\usepackage{amssymb}
\usepackage{amsmath}
\usepackage{multirow}
\usepackage{tabularx}
\usepackage{xcolor}
\usepackage{url}

\usepackage{mathtools}

\usepackage{hyperref}
\usepackage{cleveref}

\usepackage{ulem}

\newtheorem{theorem}{Theorem}[section]
\newtheorem{corollary}{Corollary}[section]
\newtheorem{proposition}{Proposition}[section]
\newtheorem{lemma}{Lemma}[section]
\newtheorem{example}{Example}[section]
\newtheorem{remark}{Remark}[section]
\newtheorem{definition}{Definition}[section]

\newenvironment{proof}[1][Proof.]{\vspace{0.5em}\textbf{#1} }{\
\hfill$\square$}

\newcommand{\F}{\mathbb{F}}

\begin{document}
\title{ $\mathbb{F}_{2}\mathbb{F}_{4}$-Additive Complementary Dual Codes}
	
%\author{Nasreddine Benbelkacem, Dipak K. Bhunia,\\ Cristina Fern\'andez-C\'ordoba, Merc\`e Villanueva}
\author{Souhila Ouagagui$^{1,2}$, Nasreddine Benbelkacem$^{2}$,\\Aicha  Batoul$^2$, Taher Abualrub$^3$}
\date{
\footnotesize $^1$  Institute of Electrical and Electronic Engineering,\\
University of M'Hamed Bougara, Boumerdes, Algeria\\
%\vspace{0.2cm}
%\footnotesize BP 32 El Alia, Bab Ezzouar, 1611 Algeries, \\
\footnotesize $^2$ Departement and Laboratory of ATN, Faculty of mathematics,\\
\footnotesize University of Science and Technology Houari Boumediene,  Algeries, Algeria\\
%\footnotesize 08193 Cerdanyola del Vall\`{e}s, Spain\\
%\footnotesize E-mails: \{Dipak.Bhunia, Cristina.Fernandez, Merce.Villanueva\}@uab.cat
\footnotesize $^3$ Department of Mathematics and Statistics, College of Arts and Sciences,\\
\footnotesize American University of Sharjah, UAE\\
}
\maketitle

\begin{abstract}
	In this paper, we investigate the structure and properties of additive
complementary dual (ACD) codes over the mixed alphabet $\mathbb{F}_{2}%
\mathbb{F}_{4}$ relative to a certain inner product defined over $\mathbb{F%
}_{2}\mathbb{F}_{4}.$ We establish sufficient conditions under which
such codes are additive complementary dual (ACD) codes. We also show
that ACD codes over $\mathbb{F}_{2}\mathbb{F}_{4}$ can be applied to
construct binary linear complementary dual codes as their images under the linear map $W$. Notably, we prove that if the binary image of a code is LCD, then the original code is necessarily ACD. An example is given where the image is a distance-optimal binary LCD code.\\

%\keywords{Hadamard code \and Gray map \and Galois Rings...}
%\subclass{MSC 94B25 \and MSC 94B60}

\end{abstract}

\section{Introduction}

Linear codes of length $n$ over the finite field $\mathbb{F}_{q},$ where $q$
is a prime power are defined to be subspaces of the $\mathbb{F}_{q}$-vector
space $\mathbb{F}_{q}^{n}.$ On the other hand, additive codes of length $n$
over the finite field $\mathbb{F}_{q},$ where $q$ is a prime power are
defined to be subgroups of $\mathbb{F}_{q}^{n}.$ Additive codes over the finite field $\mathbb{F}_4$ have been particularly introduced for their applicability in the construction of quantum error-correcting codes \cite{calderbank1998}. These codes form additive groups in $\mathbb{F}_4^n$, but they are not necessarily $\mathbb{F}_4$-linear. A major advantage of these codes lies in their capacity to meet the orthogonality constraints essential for quantum error correction. They are especially used in the design of stabilizer quantum codes, which are fundamental for ensuring the integrity of quantum information.

The concept of linear complementary dual (LCD) codes was introduced by Massey in 1992~\cite{massey}. These codes have become very important in recent years because they help protect cryptographic systems from side-channel attacks. One reason for their popularity is that they are easy to build and have a simple algebraic structure. Several works have focused on the construction of LCD codes using different techniques. For instance, a method based on orthogonal matrices was proposed to construct such codes \cite{dougherty2017combinatorics}. Another contribution focused on repeated-root constacyclic codes over finite fields, where the authors showed that LCD constacyclic and self-dual negacyclic codes exist under certain conditions \cite{benahmed2019some}. A different approach was presented in \cite{zoubir2025lcd}, where the authors proved that if an LCD code is preserved under a transitive automorphism group, then its punctured version remains LCD. They showed that BCH codes and extended Gabidulin codes are LCD, and used this result to construct new LCD codes using Paley-type bipartite graphs and punctured MacDonald codes.

Recently, a new class of codes called additive complementary dual (ACD) codes has been introduced \cite{shi2023additive}. These codes are an extension of LCD codes. They are very important in communication theory because they  preserve the useful properties of LCD codes, such as error correction and easy decoding.

Another class of codes that has received a lot of attention recently is the class of codes over mixed alphabets. This kind of study helps design more flexible codes and discover new families with better properties. For instance, in \cite{benbelkacem2020linear} and  \cite{benbelkacem2022skew}, the authors constructed the general form of generator matrices of linear codes over the mixed alphabet $\mathbb{F}_4R$, and used it to determine whether codes with specific parameters can exist. In addition, they proved a version of the MacWilliams identity and explained how these codes can be used in DNA coding systems. In a related direction, in  \cite{benbelkacem2020z2z4}, the authors analyzed the structure of ACD codes and constructed infinite families of such codes. They provided the necessary conditions for a code to be ACD and used the Gray map to transform these codes into binary LCD codes. In \cite{abualrub2020optimal}, the authors focused on cyclic codes over the mixed alphabet $\mathbb{F}_2 \mathbb{F}_4$ and constructed several classes of codes with good parameters as an application of their results.

In this paper, we continue the study of additive codes over $\mathbb{F}_2\mathbb{F}_4$, focusing on constructions and properties derived from generator matrices. We establish sufficient conditions under which these codes satisfy the ACD property. Furthermore, we show
that ACD codes over $\mathbb{F}_{2}\mathbb{F}_{4}$ can be applied to construct LCD codes as their images under the linear map $W$.

This paper is organized as follows. In Section \ref{sec2:per}, we present the general form of the generator matrix of an $\mathbb{F}_{2}\mathbb{F}_{4}$-additive code, as well as its corresponding type. In Section \ref{sec3:per}, we investigate key properties of $\mathbb{F}_{2}\mathbb{F}_{4}$-additive codes and demonstrate that the dual of any $\mathbb{F}_{2}\mathbb{F}_{4}$-additive code is itself an $\mathbb{F}_{2}\mathbb{F}_{4}$-additive code. In Section \ref{sec4:per}, we provide several results and establish sufficient conditions under which such codes are ACD. 
In Section \ref{sec5:per}, we study various cases of ACD codes by analyzing whether their punctured codes are complementary dual. Moreover, we provide a counterexample to show that even when both punctured codes are complementary dual, the code $C$ is not necessarily an ACD code. In Section \ref{sec6:per}, we show that ACD codes over $\mathbb{F}_{2}\mathbb{F}_{4}$ can be applied to construct LCD codes as their images under the linear map $W$. Moreover, we establish an important result showing that if the binary image of an $\mathbb{F}_2\mathbb{F}_4$-additive code is an LCD code, then the original code must necessarily be an ACD code. We also present an example of an ACD code whose binary image is a distance-optimal binary LCD code. Finally, Section \ref{sec7:per} presents the conclusion and future work.
 
\section{Preliminaries}\label{sec2:per}

Let $\mathbb{F}_{2}=\left\{ 0,1\right\}$ be the binary finite field and $%
\mathbb{F}_{4}=\mathbb{F}_{2}\left[ x\right] /\left\langle
x^{2}+x+1\right\rangle =\left\{ 0,1,\omega ,\omega ^{2}=\omega +1\right\} $
be the finite field of $4$ elements such that $\omega ^{2}+\omega +1=0.$ A
linear code of length $\alpha $ over $\mathbb{F}_{2}$ is a subspace of $%
\mathbb{F}_{2}^{\alpha }$ and an additive code of length $\beta $ over $%
\mathbb{F}_{4}$ is a subgroup of the additive group $\mathbb{F}_{4}^{\beta
}.$ A linear code $C_{1}$ of length $\alpha $ over $\mathbb{F}%
_{2}$ will have $\left\vert C_{1}\right\vert =2^{k_{1}},$ where $\dim \left(
C_{1}\right) =k_{1},$ with $0\leq k_{1}\leq \alpha ,$ while an additive code
of length $\beta $ over $\mathbb{F}_{4}$ will have $\left\vert
C_{2}\right\vert =2^{k_{2}},$ where $0\leq k_{2}\leq 2\beta .$

\begin{definition}
A non-empty subset $C$ of $\mathbb{F}_{2}^{\alpha }\mathbb{F}_{4}^{\beta }$
is called an $\mathbb{F}_{2}\mathbb{F}_{4}$-additive code of length $%
n=\alpha +\beta $ if $C$ is a subgroup of $\mathbb{F}_{2}^{\alpha }\mathbb{F}%
_{4}^{\beta}.$
\end{definition}

This class of codes was introduced in \cite{abualrub2020optimal} and it is an
example of codes over mixed alphabets similar to codes over $\mathbb{Z}_{2}%
\mathbb{Z}_{4}$ and $\mathbb{Z}_{2}\mathbb{Z}_{2^{s}}$. 
\cite{borges2010linear}, \cite{fernandez2010linear}, \cite{aydougdu2013structure}. 

Notice that the set $\mathbb{F}_{2}^{\alpha }\mathbb{F}_{4}^{\beta }$ is
also a vector space over the finite field $\mathbb{F}_{2}$ under the standard addition and scalar multiplication $\pmod{2}$. Hence, $\mathbb{F}%
_{2}\mathbb{F}_{4}$-additive codes of length $n=\alpha +\beta $ are also $%
\mathbb{F}_{2}$-subspaces of $\mathbb{F}_{2}^{\alpha }\mathbb{F}_{4}^{\beta}.$ 

Let $X$ be the set of the first $\alpha $ coordinate positions over $\mathbb{%
F}_{2}$, and $Y$ be the set of the last $\beta $ coordinate positions over $%
\mathbb{F}_{4}.$ Denote by $C_{X}$ and $C_{Y}$ the punctured codes obtained
from $C$ by removing the coordinates not in $X$ and $Y$, respectively. I.e., 
$C_{X}=\left\{ c_{1}:\left( c_{1}|c_{2}\right) \in C\right\} $ and $%
C_{Y}=\left\{ c_{2}:\left( c_{1}|c_{2}\right) \in C\right\} .$ Moreover, let 
$C_{b}$ be the subcode of $C$ which contains all the binary codewords in $C$
and let $k_{2}^{\prime \prime }$ be the dimension of $\left( C_{b}\right)
_{X}$ which is a binary linear code. Based on this discussion, it is clear
that $C$ is isomorphic to $\mathbb{F}_{2}$-subspace of the form $\mathbb{F}%
_{2}^{k_{1}}\times \mathbb{F}_{2}^{k_{2}^{\prime \prime }}\times \mathbb{F}%
_{2}^{2k_{2}^{\prime }}$ and $\left\vert C\right\vert
=2^{k_{1}+k_{2}^{\prime \prime }+2k_{2}^{\prime }},$ where $k_{1}\leq \alpha 
$ and $k_{2}=k_{2}^{\prime \prime }+2k_{2}^{\prime }\leq \beta .$ We will
say that $C$ is an additive code of type $\left( \alpha ,\beta ;k_{1},k_{2}^{\prime
},k_{2}^{\prime \prime }\right).$ 
Note that $C_{X}$ is a binary linear code of type $\left( \alpha ,0 ;k_{1},0,0\right),$ and $C_{Y}$ is a quaternary additive code of type $\left( 0 ,\beta ;0,k_{2}^{\prime
},k_{2}^{\prime \prime }\right).$  

Any element $z$ in $\mathbb{F}_{4}$ can be written uniquely as $z=b+\omega q$,
where $b$ and $q$ are binary numbers. The mapping $W$ given by: 
\begin{eqnarray*}
W &:&\mathbb{F}_{2}^{\alpha }\mathbb{F}_{4}^{\beta }\rightarrow \mathbb{F}%
_{2}^{\alpha +2\beta } \\
&&W\left( a_{0},a_{1},\ldots ,a_{\alpha-1}|\left( b_{0}+\omega q_{0}\right),
\left( b_{1}+\omega q_{1}\right), \ldots ,\left( b_{\beta-1}+\omega q_{\beta-1}\right)
\right) \\
&=&\left( a_{0},a_{1},\ldots ,a_{\alpha-1}|\left( b_{0}+q_{0}\right), \ldots ,\left(
b_{\beta-1}+q_{\beta-1}\right) |q_{0},\ldots ,q_{\beta-1}\right),
\end{eqnarray*}%
was defined in \cite{Abualrub2020optimal}, where it was also proved that $W$ is a linear bijection.

The next result provides a generator matrix for additive codes over $\mathbb{%
F}_{2}\mathbb{F}_{4}.$

\begin{lemma}
Let $C$ be an $\mathbb{F}_{2}\mathbb{F}_{4}$-additive code of length $%
n=\alpha +\beta ,$ then $C$ is equivalent to an $\mathbb{F}_{2}\mathbb{F}%
_{4}$-additive code of length $n=\alpha +\beta $ with a generator matrix of
the form:%
\begin{equation*}
G=\left( 
\begin{tabular}{ll|llll}
$I_{k_{1}}$ & $A_{1}$ & $0$ & $0$ & $\omega B_{1}$ &  \\ \hline
$0$ & $C_{1}$ & $I_{k_{2}^{\prime }}$ & $\omega D_{1}$ & $D_{2}$ &  \\ 
$0$ & $0$ & $\omega I_{k_{2}^{\prime }}$ & $\omega D_{3}$ & $D_{4}$ &  \\ 
$0$ & $0$ & $0$ & $I_{k_{2}^{\prime \prime }}$ & $D_{5}$ & 
\end{tabular}%
\right),
\end{equation*}%
where $I_{k_1}$, $I_{k_2'}$, and $I_{k_2''}$ are identity matrices with $k_2 = 2k_2' + k_2''$. The matrices $D_2$ and $D_4$ are defined over $\mathbb{F}_4$, while $A_1$, $B_1$, $C_1$, $D_1$, $D_3$, and $D_5$ are all binary matrices.
\end{lemma}

\begin{proof}
We know that any linear code over $\mathbb{F}_{2}$ is a permutation
equivalent to a code with a generator matrix of the form $A=\left(
I_{k_{1}}|A_{1}\right) ,$ where $A_{1}$ is a binary matrix over $\mathbb{F}%
_{2}.$ In \cite{calderbank1998} showed that any additive code over $\mathbb{F%
}_{4}$ has a generator matrix of the form:\\ $$D=\left( 
\begin{tabular}{llll}
$I_{k_{2}^{\prime }}$ & $\omega D_{1}$ & $D_{2}$ &  \\ 
$\omega I_{k_{2}^{\prime }}$ & $\omega D_{3}$ & $D_{4}$ &  \\ 
$0$ & $I_{k_{2}^{\prime \prime }}$ & $D_{5}$ &  \\ 
&  &  & 
\end{tabular}%
\right) .$$ \\As a result of these two matrices and arranging the matrix $A$ in
the first block and the matrix $D$ in the last block while completing the
remaining blocks with appropriate matrices and applying row and column
operations, we get that $C$ has a generator matrix of the form: $$G=\left( 
\begin{tabular}{ll|llll}
$I_{k_{1}}$ & $A_{1}$ & $0$ & $0$ & $\omega B_{1}$ &  \\ \hline
$0$ & $C_{1}$ & $I_{k_{2}^{\prime }}$ & $\omega D_{1}$ & $D_{2}$ &  \\ 
$0$ & $0$ & $\omega I_{k_{2}^{\prime }}$ & $\omega D_{3}$ & $D_{4}$ &  \\ 
$0$ & $0$ & $0$ & $I_{k_{2}^{\prime \prime }}$ & $D_{5}$ & 
\end{tabular}%
\right),$$where $I_{k_1}$, $I_{k_2'}$, and $I_{k_2''}$ are identity matrices with $k_2 = 2k_2' + k_2''$. The matrices $D_2$ and $D_4$ are defined over $\mathbb{F}_4$, while $A_1$, $B_1$, $C_1$, $D_1$, $D_3$, and $D_5$ are all binary matrices.
\end{proof}
\section{Duality of $\mathbb{F}_{2}\mathbb{F}_{4}$-additive codes} \label{sec3:per}

In this section, we introduce the standard inner product over $\F_2^\alpha\F_4^\beta$, as well as the corresponding dual code, focusing on some of their properties.\\
Let $u=\left( a_{0},a_{1},\ldots ,a_{\alpha -1}|b_{0},b_{1},\ldots ,b_{\beta
-1}\right)$,
$v=\left( c_{0},c_{1},\ldots ,c_{\alpha -1}|d_{0},d_{1},\ldots ,d_{\beta -1}\right)$
$\in \mathbb{F}_{2}^{\alpha }\mathbb{F}_{4}^{\beta }$.  The standard inner product on $\mathbb{F}_{2}^{\alpha }\mathbb{F}_{4}^{\beta }$ is defined by:

\begin{equation*}
\left\langle u,v\right\rangle _{4}=\left( \omega \sum_{i=0}^{\alpha
-1}a_{i}c_{i}+\sum_{j=0}^{\beta -1}b_{j}d_{j}\right) \in \mathbb{F}_{4}.
\end{equation*}

Let $C$ be an $\mathbb{F}_{2}\mathbb{F}_{4}$-additive code of length $%
n=\alpha +\beta $ and consider the punctured codes $C_{X}$ and $C_{Y}$ with
generator matrices $G_{X}$ and $G_{Y},$ respectively. Then $G=\left(
G_{X}|G_{Y}\right) $ is a generator matrix for $C.$ Notice that $$%
GG^{t}=\left( G_{X}|G_{Y}\right) \left( 
\begin{array}{c}
G_{X}^{t} \\ 
G_{Y}^{t}%
\end{array}%
\right) =\omega G_{X}G_{X}^{t}+G_{Y}G_{Y}^{t}\in \mathbb{F}_{4},$$ based on
the inner product defined above.

\begin{definition}
Let $C$ be an $\mathbb{F}_{2}\mathbb{F}_{4}$-additive code of length $%
n=\alpha +\beta.$ The dual code of $C$ is defined as $C^{\perp }=\left\{
v\in \mathbb{F}_{2}^{\alpha }\mathbb{F}_{4}^{\beta } \mid\left\langle
v,u\right\rangle _{4}=0,~\forall u\in C~\right\} .$
\end{definition}

\begin{lemma}
Let $C$ be an $\mathbb{F}_{2}\mathbb{F}_{4}$-additive code of length $%
n=\alpha +\beta .$ Then $C^{\perp }$ is an  $\mathbb{F}_{2}\mathbb{F}%
_{4}$-additive code of length $n=\alpha +\beta .$
\end{lemma}

\begin{proof}
To prove that $C^{\perp }$ is an  $\mathbb{F}_{2}\mathbb{F}_{4}$%
-additive code of length $n=\alpha +\beta ,$ we only need to prove that $%
C^{\perp }$ is an $\mathbb{F}_{2}$-subspace of $\mathbb{F}_{2}^{\alpha }%
\mathbb{F}_{4}^{\beta }.$ Let $u,v\in C^{\perp }$ and let $w\in C.$ Then $\left\langle u,w\right\rangle
_{4}=\left\langle v,w\right\rangle _{4}=0$ and $\left\langle
u-v,w\right\rangle _{4}=\left\langle u,w\right\rangle _{4}-\left\langle
v,w\right\rangle _{4}=0.$ Therefore, $C^{\perp }$ is an additive code.
\end{proof}

It is also worth mentioning that $C^{\perp }$ is not a linear code over $%
\mathbb{F}_{4}$ because if $v\in C^{\perp },$ then $\omega v$ is not defined
and hence it is not in $C^{\perp }.$\\

We now define the duals of ${C_X}$ and $C_Y$ with respect to the Euclidean inner products on $\mathbb{F}_2^\alpha$ and $\mathbb{F}_4^\beta$, respectively. The dual of ${C_X}$ is
$${C_X}^\perp=\left\{
c\in \mathbb{F}_{2}^{\alpha } \mid\left\langle
c,a\right\rangle =0,~\forall a\in {C_X}~\right\},$$ and the dual of ${C_Y}$ is $${C_Y}^\perp=\left\{
d\in \mathbb{F}_{4}^{\beta } \mid\left\langle
d,b\right\rangle =0,~\forall b\in {C_Y}~\right\}.$$ 

As previously noted, $C_{X}$ is a binary linear code of length $\alpha$, and $C_{Y}$ is a quaternary additive code of length $\beta $. The following lemma provides a description of their duals.

\begin{lemma}
Let $C$ be an $\mathbb{F}_{2}\mathbb{F}_{4}$-additive code of length $%
n=\alpha +\beta .$ Then $C_{X}^{\perp }$ is a binary linear code of length $%
\alpha $, while the code $C_{Y}^{\perp }$ is a linear code of length $\beta $
over $\mathbb{F}_{4}.$
\end{lemma}

\begin{proof} The case that $C_{X}^{\perp }$ is a binary linear code of length $\alpha $
is clear. We will show that $C_{Y}^{\perp }$ is a linear code of length $\beta 
$ over $\mathbb{F}_{4}.$ Let $d_{1}, d_{2}\in C_{Y}^{\perp }$ and let $b\in C_{Y}.$
Then $\left\langle d_{1},b\right\rangle =\left\langle d_{2},b\right\rangle =0$ and $%
\left\langle d_{1}-d_{2},b\right\rangle =\left\langle d_{1},b\right\rangle -\left\langle
d_{2},b\right\rangle =0.$ Moreover, we have $\left\langle zd_{1},b\right\rangle
=z\left\langle d_{1},b\right\rangle =0$ for any $z\in \mathbb{F}_{4}.$
Therefore, $C_{Y}^{\perp }$ is a linear code of length $\beta $ over $%
\mathbb{F}_{4}.$
\end{proof}

\begin{definition}
An additive code $C$ over $\mathbb{F}_{2}\mathbb{F}_{4}$ is called separable
if ${\mathcal{C}}={\mathcal{C}}_{X}\times {\mathcal{\ C}}_{Y}.$    
\end{definition}

\begin{lemma}
\label{lemma 7}
Suppose that ${\mathcal{C}}={\mathcal{C}}_{X}\times {\mathcal{\ C}}_{Y}$ is
a separable additive code over $\mathbb{F}_{2}\mathbb{F}_{4}.$ Then ${%
\mathcal{C}}_{X}^{\perp }\times {\mathcal{\ C}}_{Y}^{\perp }\subseteq
C^{\perp }.$
\end{lemma}

\begin{proof}
Let $v=\left( c|d\right) \in {\mathcal{C}}_{X}^{\perp }\times {\mathcal{\ C}}%
_{Y}^{\perp }$ and $u=\left( a|b\right) \in C.$ Then $c\in {\mathcal{C}}%
_{X}^{\perp },~a\in C_{X},$ $d\in {\mathcal{C}}_{Y}^{\perp }$ and $b\in
C_{Y}.$ Thus, $\sum\limits_{i=0}^{\alpha-1 }c_{i}a_{i}=0$ in $\mathbb{F}_{2}$
and $\sum\limits_{j=0}^{\beta-1 }d_{j}b_{j}=0$ in $\mathbb{F}_{4}.$ Hence, $%
\left\langle v,u\right\rangle _{4}=\left( \omega \sum\limits_{i=0}^{\alpha
-1}c_{i}a_{i}+\sum\limits_{j=0}^{\beta -1}d_{j}b_{j}\right) =0$ and $v\in C^{\perp
}.$ Therefore, ${\mathcal{C}}_{X}^{\perp }\times {\mathcal{\ C}}_{Y}^{\perp
}\subseteq C^{\perp }.$
\end{proof}

The following key lemma arises as a direct application of the mapping $W$.

\begin{lemma}
\label{W(Dual-Code)} Let $C$ be an $\mathbb{F}_{2}\mathbb{F}_{4}$-additive
code of length $n=\alpha +\beta.$ Then, $$W\left(C^{\perp }\right) \subseteq
W\left(C\right)^{\perp }.$$
\end{lemma}

\begin{proof}
Suppose that $v=\left( a|c+\omega q\right) \in C^{\perp }$ and $w=W\left(
v\right) =\left( a|c+q|q\right) \in W\left( C^{\perp }\right) $. Let $%
w^{\prime }=W\left( v^{\prime }\right) =\left( a^{\prime }|c^{\prime
}+q^{\prime }|q^{\prime }\right) \in W\left( C\right) ,$ where $v^{\prime
}=\left( a^{\prime }|b^{\prime }+\omega q^{\prime }\right) \in C.$ Then%
\begin{equation*}
\left\langle v,v^{\prime }\right\rangle _{4}=cc^{\prime }+qq^{\prime
}+\omega \left( aa^{\prime }+cq^{\prime }+qc^{\prime }+qq^{\prime }\right)
=0.
\end{equation*}%
Hence, $cc^{\prime }+qq^{\prime }=0$ and $aa^{\prime }+cq^{\prime
}+qc^{\prime }+qq^{\prime }=0.$ Thus, $aa^{\prime }+cq^{\prime }+qc^{\prime
}+cc^{\prime }=0$ and $w\in W\left( C\right) ^{\perp }.$
Therefore, $W\left( C^{\perp }\right) \subseteq W\left( C\right) ^{\perp }.$
\end{proof}

The next example shows that $W\left( C^{\perp }\right) \neq W\left( C\right)
^{\perp }$ for such code $C.$

\begin{example}
Let $C$ be an $\mathbb{F}_{2}\mathbb{F}_{4}$-additive code of type $\left( 2 ,2 ;0,1,0 \right)$, with
generator matrix: $$G=\left( 
\begin{tabular}{ll|ll}
$1$ & $1$ & $\omega $ & $1$ \\ 
$0$ & $1$ & $\omega^2 $ & $\omega $%
\end{tabular}%
\right).$$ Then $C^{\perp }$ has a generator matrix given by: $$\left( 
\begin{tabular}{ll|ll}
$0$ & $0$ & $1$ & $\omega $ \\ 
$0$ & $0$ & $\omega $ & $\omega ^{2}$%
\end{tabular}%
\right).$$ The generator matrix for $W\left( C^{\perp }\right) $ is: $$\left( 
\begin{tabular}{llllll}
$0$ & $0$ & $1$ & $1$ & $0$ & $1$ \\ 
$0$ & $0$ & $1$ & $0$ & $1$ & $1$%
\end{tabular}%
\right),$$ and $\left\vert W\left( C^{\perp }\right) \right\vert =4.$ We also
have $W\left( C\right) $ is a binary linear code of length $6$ and dimension $2$
with generator matrix given by: $$W\left( G\right) =\left( 
\begin{tabular}{llllll}
$1$ & $1$ & $1$ & $1$ & $1$ & $0$ \\ 
$0$ & $1$ & $0$ & $1$ & $1$ & $1$%
\end{tabular}%
\right) .$$ Thus, $\left\vert W\left( C\right) \right\vert =4$ and $%
\left\vert W\left( C\right) ^{\perp }\right\vert =\dfrac{2^{6}}{4}=16.$
Therefore, $W\left( C^{\perp }\right) \neq W\left( C\right) ^{\perp }.$
\end{example}

\begin{corollary}
Let $C$ be an $\mathbb{F}_{2}\mathbb{F}_{4}$-additive code of length $%
n=\alpha +\beta .$ Then, $$\left\vert C\right\vert \left\vert C^{\perp
}\right\vert \leq 2^{\alpha +2\beta }.$$
\end{corollary}

\begin{proof}
Since the map $W$ is a bijection, then $\left\vert C\right\vert =\left\vert
W\left( C\right) \right\vert $ and $\left\vert C^{\perp }\right\vert
=\left\vert W\left( C^{\perp }\right) \right\vert .$ Apply Lemma \ref%
{W(Dual-Code)}, we get that $\left\vert C\right\vert \left\vert C^{\perp
}\right\vert =\left\vert W\left( C\right) \right\vert \left\vert W\left(
C^{\perp }\right) \right\vert \leq \left\vert W\left( C\right) \right\vert
\left\vert W\left( C\right) ^{\perp }\right\vert =2^{\alpha +2\beta }.$    
\end{proof}

\section{Additive complementary dual codes}
\label{sec4:per}
We proceed in this section to analyze the properties of additive complementary dual (ACD) codes, using the inner product introduced in the preceding section.

\begin{definition} An additive code $C\subseteq \mathbb{F}_{2}^{\alpha }\mathbb{F}_{4}^{\beta }$ is called an additive complementary dual (ACD) if $ C \cap  C^{\perp} = \mathbf{\{0\}}$.
\end{definition}

In the special case where $\beta = 0$, the code $C$ is an entirely binary $LCD$ code. Similarly, when $\alpha = 0$, the code $C$ becomes a purely quaternary ACD code.

The following proposition gives a characterization of a Euclidean $LCD$ code.
\begin{proposition} \cite{massey}
\label{LCD codes}
Let $G$ be a generator matrix for an $[n,k]$ linear code $C$ over a field $\mathbb{F}_{q}$. Then $C$ is a Euclidean LCD code if and only if the $k \times k$ matrix $G.  G^{t}$ is invertible (i.e., nonsingular).
\end{proposition}

We consider the following three cases. If any one of them is satisfied, then the code $C$ over $\mathbb{F}_{2}\mathbb{F}_{4}$ qualifies as an additive complementary dual (ACD) code, as established in the following theorem.

\begin{theorem}
\label{ACD-Three-Cases} Let $G$ be a generator matrix for an additive code $C$
over $\mathbb{F}_{2}\mathbb{F}_{4}$ and let $\left\{ V_{1},V_{2},\ldots
,V_{l}\right\} $ be the set of rows of $G.$ Suppose any case of the
following cases occur for all $i,j=1,2,\ldots ,l$:

Case I: $\left\langle V_{i},V_{j}\right\rangle \in \mathbb{F}_{2}$ for $%
i\neq j$ and $\left\langle V_{i},V_{i}\right\rangle \notin \mathbb{F}_{2}.$

Case II: $\left\langle V_{i},V_{j}\right\rangle \in \left\{ 0,\omega
\right\} $ for $i\neq j$ and $\left\langle V_{i},V_{i}\right\rangle \notin
\left\{ 0,\omega \right\} .$

Case III: $\left\langle V_{i},V_{j}\right\rangle \in \left\{ 0,\omega
^{2}\right\} $ for $i\neq j$ and $\left\langle V_{i},V_{i}\right\rangle
\notin \left\{ 0,\omega ^{2}\right\}.$\\
Then $C$ is an ACD code.
\end{theorem}

\begin{proof}
We will prove the first case. The other cases are proven similarly.\\ Let $%
0\neq x=\sum\limits_{i=1}^{l}\lambda _{i}V_{i}\in C,$ where $\lambda _{i}\in 
\mathbb{F}_{2}$ for all $i=1,2,\ldots ,l.$ Then, $$\left\langle
x,V_{j}\right\rangle =\lambda _{i}\sum\limits_{\substack{ i=1  \\ i\neq j}}%
^{l}\left\langle V_{i},V_{j}\right\rangle +\lambda _{j}\left\langle
V_{j},V_{j}\right\rangle .$$

Case I: Suppose that $\left\langle V_{i},V_{j}\right\rangle \in \mathbb{F}%
_{2}$ for $i\neq j$ and $\left\langle V_{i},V_{i}\right\rangle \notin 
\mathbb{F}_{2}.$ Then $\lambda _{i}\sum\limits_{\substack{ i=1  \\ i\neq j}}%
^{l}\left\langle V_{i},V_{j}\right\rangle =0$ or $\lambda _{i}\sum\limits 
_{\substack{ i=1  \\ i\neq j}}^{l}\left\langle V_{i},V_{j}\right\rangle =1.$
Suppose that $\lambda _{i}\sum\limits_{\substack{ i=1  \\ i\neq j}}%
^{l}\left\langle V_{i},V_{j}\right\rangle =0.$ Since $x\neq 0,$ then there
is at least one $j$ such that $\lambda _{j}\neq 0$ and $\left\langle
x,V_{j}\right\rangle =\lambda _{i}\sum\limits_{\substack{ i=1  \\ i\neq j}}%
^{l}\left\langle V_{i},V_{j}\right\rangle +\lambda _{j}\left\langle
V_{j},V_{j}\right\rangle =\left\langle V_{j},V_{j}\right\rangle \neq 0.$
Thus, $x\notin C^{\perp }$ and $C$ is ACD code. If $\lambda _{i}\sum\limits 
_{\substack{ i=1  \\ i\neq j}}^{l}\left\langle V_{i},V_{j}\right\rangle =1,$
then $\left\langle x,V_{j}\right\rangle =\lambda _{i}\sum\limits_{\substack{ %
i=1  \\ i\neq j}}^{l}\left\langle V_{i},V_{j}\right\rangle +\lambda
_{j}\left\langle V_{j},V_{j}\right\rangle\in\{\omega ,\omega
^{2}\} $ and $x\notin C^{\perp }.$
\end{proof}

We now present an example of an ACD code corresponding to each of the three cases discussed above.

\begin{example}
Let $C$ be an $\mathbb{F}_{2}\mathbb{F}_{4}$-additive code of type $\left( 2 ,2 ;0,1,0 \right)$ generated by:
 $$G=\left( 
\begin{tabular}{ll|ll}
$1$ & $1$ & $\omega $ & $1$ \\ 
$0$ & $1$ & $\omega^{2} $ & $\omega $%
\end{tabular}%
\right).$$ Then $\left\langle V_{1},V_{2}\right\rangle _{4}=1
,~\left\langle V_{1},V_{1}\right\rangle _{4}=\omega $ and $\left\langle
V_{2},V_{2}\right\rangle _{4}=\omega ^{2}.$ Thus, this code satisfies case I
and is an ACD code.
\end{example}

\begin{example}
Let $C$ be an $\mathbb{F}_{2}\mathbb{F}_{4}$-additive code of type $\left( 2 ,1 ;1,0,1 \right)$, with
generator matrix:
 $$G=\left( 
\begin{tabular}{ll|l}
$1$ & $0$ & $1$ \\ 
$0$ & $1$ & $\omega $%
\end{tabular}%
\right).$$ We find that $\left\langle V_{1},V_{2}\right\rangle _{4}=\omega
,~\left\langle V_{1},V_{1}\right\rangle _{4}=\omega ^{2}$ and $\left\langle
V_{2},V_{2}\right\rangle _{4}=1,$ which corresponds to case II. Therefore, the code is an ACD code.
\end{example}

\begin{example}
Let $C$ be an $\mathbb{F}_{2}\mathbb{F}_{4}$-additive code of type $\left( 4, 2; 0, 1, 0 \right)$, with
generator matrix:
 $$G=\left( 
\begin{tabular}{llll|ll}
$1$ & $0$ & $1$ & $0$ & $1$ & $\omega $ \\ 
$0$ & $1$ & $0$ & $1$ & $\omega $ & $\omega ^{2}$%
\end{tabular}%
\right).$$ Then $\left\langle V_{1},V_{2}\right\rangle _{4}=\omega
^{2},~\left\langle V_{1},V_{1}\right\rangle _{4}=\omega $ and $\left\langle
V_{2},V_{2}\right\rangle _{4}=1.$ This code satisfies the last case and is an ACD code.
\end{example}

We provide below an example of an ACD code that does not satisfy any of the above three cases.

\begin{example}
Let $C$ be an $\mathbb{F}_{2}\mathbb{F}_{4}$-additive code of type $\left( 2 ,3 ;1,1,0 \right)$ generated by:
 $$G=\left( 
\begin{tabular}{ll|lll}
$1$ & $0$ & $0$ & $0$ & $\omega $ \\\hline

$0$ & $1$ & $1$ & $\omega $ & $0$ \\ 
$0$ & $0$ & $\omega $ & $\omega $ & $0$%
\end{tabular}%
\right).$$ Then $\left\langle V_{1},V_{2}\right\rangle _{4}=\left\langle
V_{1},V_{3}\right\rangle _{4}=0,~\left\langle V_{2},V_{3}\right\rangle
_{4}=\left\langle V_{1},V_{1}\right\rangle _{4}=1$, and $\left\langle
V_{2},V_{2}\right\rangle _{4}=\left\langle V_{3},V_{3}\right\rangle _{4}=0.$
Thus, $\left\langle V_{i},V_{i}\right\rangle _{4}\in \mathbb{F}_{2}$ for all 
$i=1,2,3$ and the code does not satisfy any of the three cases in Theorem \ref%
{ACD-Three-Cases}, even though the code is ACD code.
\end{example}

\begin{corollary}
\label{coro- three-cases}
Let $G$ be a generator matrix for an additive code $C$ over $\mathbb{F}_{2}%
\mathbb{F}_{4}$ and suppose that $GG^{t}=\left( a_{ij}\right) ,$ where the
entries $a_{ij}$ satisfy any of the cases in Theorem \ref{ACD-Three-Cases}%
. Then $C$ is ACD code.
\end{corollary}

\begin{proof}
The proof is a direct application of Theorem \ref{ACD-Three-Cases}.
\end{proof}

\begin{theorem}
\label{ACD-Three-Cases-Similar}Let $G$ be a generator matrix for an additive
code $C$ over $\mathbb{F}_{2}\mathbb{F}_{4}$ and let $\left\{
V_{1},V_{2},\ldots ,V_{l}\right\} $ be the set of rows of $G.$ Suppose that $%
\left\langle V_{i},V_{j}\right\rangle =0$ if $i\neq j$ and $\left\langle
V_{i},V_{i}\right\rangle \neq 0$ for all $i,j=1,2,\ldots ,l$. Then $C$ is an
ACD code.
\end{theorem}

\begin{proof}
The proof is similar to the proof of Theorem \ref{ACD-Three-Cases}.
\end{proof}

\begin{example}
Let $C$ be an $\mathbb{F}_{2}\mathbb{F}_{4}$-additive code of type $\left( 2 ,3 ;0,1,0 \right)$, with
generator matrix:
 $$G=\left( 
\begin{tabular}{ll|lll}
$1$ & $0$ & $\omega $ & $\omega ^{2}$ & $1$ \\ 
$0$ & $1$ & $\omega ^{2}$ & $1$ & $\omega $%
\end{tabular}%
\right) .$$ It is clear that $C$ satisfies the conditions of Theorem \ref%
{ACD-Three-Cases-Similar} and hence it is an ACD code.
\end{example}

The next Theorem provides a very important class of ACD codes based on their
generator matrices.

\begin{theorem}
\label{invertible-LCD}Let $G$ be a generator matrix for an additive code $C$
over $\mathbb{F}_{2}\mathbb{F}_{4}$ and suppose that $GG^{t}$ is an
invertible diagonal matrix. Then $C$ is an ACD code.
\end{theorem}

\begin{proof}
Let $\left\{ V_{1},V_{2},\ldots ,V_{l}\right\} $ be the set of rows of $G.$
Since $GG^{t}$ is a diagonal matrix, then $\left\langle
V_{i},V_{j}\right\rangle =0$ if $i\neq j$ and $ det( GG^{t})
=\prod\limits_{i=1}^{l}\left\langle V_{i},V_{i}\right\rangle \neq
0\Leftrightarrow \left\langle V_{i},V_{i}\right\rangle \neq 0,$ for all $%
i=1,2,...,l.$ Therefore, by Theorem \ref{ACD-Three-Cases-Similar}, $C$ is an
ACD code.
\end{proof}

The next example is an application of Theorem \ref{invertible-LCD}.

\begin{example}
Let $C$ be an $\mathbb{F}_{2}\mathbb{F}_{4}$-additive code of type $\left( 2 ,1 ;1,0,1 \right)$ generated by: 
$$G=\left( 
\begin{tabular}{ll|l}
$1$ & $1$ & $1$ \\ 
$0$ & $1$ & $\omega $%
\end{tabular}%
\right) .$$ Then $GG^{t}=\left( 
\begin{tabular}{ll}
$1$ & $0$ \\ 
$0$ & $1$%
\end{tabular}%
\right) $ satisfies the conditions in Theorem \ref{invertible-LCD} and $C$
is an ACD code.
\end{example}

\section{Complementary duality of $C,~C_{X}$ and $C_{Y}$} \label{sec5:per}

In this section, we examine the relationship between an ACD code $C$ over $%
\mathbb{F}_{2}\mathbb{F}_{4}$ and its punctured codes $C_{X}$ and $C_{Y}$. In particular, we explore the conditions under which the code $C$ is an additive complementary dual (ACD), assuming that only one of its punctured codes, $C_X$ or $ C_Y$, satisfies the LCD or ACD property, respectively. Furthermore, we present a counterexample to show that the ACD property is not necessarily preserved, even when $C_X$ is an LCD code and $C_Y$ is an ACD code.

\begin{theorem}
\label{th23: span}
Let $C$ be an $\F_2 \F_4$-additive code such that $C_X$ is a binary LCD code and $C_Y$ is an $\F_4$-additive self-orthogonal code. If the binary part $a$ of every nonzero codeword $u = (a , b)\in C $ is nonzero, then  $C$ is an ACD code.
\end{theorem}

\begin{proof} Assume that the code $C$ is not an ACD code. Therefore, there exists a nonzero codeword $u \in C \cap C^{\perp}$. Let $u = (a , b)\in C $, where $ a \in C_X$ and $ b \in C_Y$. Since $u \in C^{\perp}$, it follows that $\left\langle u,v\right\rangle _{4} =  \omega \sum\limits_{i=0}^{\alpha
-1}a_{i}c_{i}+\sum\limits_{j=0}^{\beta -1}b_{j}d_{j} = 0$ for all $v = (c , d)\in C$. Since $C_Y$ is an $\F_4$-additive  self-orthogonal code, it follows that $\left\langle u,v\right\rangle _{4} = \omega \sum\limits_{i=0}^{\alpha
-1}a_{i}c_{i} = 0$. This implies that $\sum\limits_{i=0}^{\alpha
-1}a_{i}c_{i} = 0$. Thus, $\left\langle a,c\right\rangle _{2} = 0$ for all $ c \in C_X $. Hence, $a \in C_X^\perp$. Since the binary part $a$ is nonzero, this contradicts the condition that $ C_X $ is an LCD code. Therefore, we conclude that $C$ is an ACD code.
\end{proof}

\begin{example} \label{exa11: span}
Let $C$ be an $\mathbb{F}_{2}\mathbb{F}_{4}$-additive code of type $\left( 2 ,2 ;1,0,1 \right)$ generated by:
$$
G=\left(\begin{matrix}
 1 &  1 &   \bigm| & 1 & 1 \\
   0 & 1 &  \bigm| & \omega & \omega
\end{matrix}\right)
.$$
We have $C_X= \{ (0, 0), (1, 1),(0,1), (1,0 )\}$, and $C_Y= \{ ( 0, 0), ( 1, 1),( \omega, \omega), ( {\omega}^2, {\omega}^2)\}$. We observe that the codes $C_X$ and $C_Y$ satisfy the conditions stated in Theorem \ref{th23: span},  and that the binary part of every nonzero codeword in $C$ is nonzero. Consequently, $C$ is an ACD code.
\end{example}

As a direct consequence of Theorem \ref{th23: span}, we obtain the following corollary.

\begin{corollary}
\label{cor25: span}
Let $C$ be an $\F_2 \F_4$-additive code. If $C_X$ is not a binary LCD code and $C_Y$ is an $\F_4$-additive self-orthogonal code. Then $C$ is not an ACD code.
\end{corollary}

\begin{proof} The proof is an immediate consequence of  Theorem \ref{th23: span}.
\end{proof}

\begin{example}
Let $C$ be an $\mathbb{F}_{2}\mathbb{F}_{4}$-additive code of type $\left( 3 ,2 ;1,0,1 \right)$ generated by:
$$
G=\left(\begin{matrix}
 1 & 0 &  1 &   \bigm| & 1 & 1 \\
   0 & 1 & 0 &  \bigm| & \omega & \omega
\end{matrix}\right)
.$$
In this example, we observe that $C_X$ is not a binary LCD code and $C_Y$ is an $\F_4$-additive self-orthogonal code. As a consequence of Corollary \ref{cor25: span}, the code $C$ is not an ACD code.
\end{example}

\begin{proposition}
\label{th28: span}
Let $C$ be an $\mathbb{F}_{2}\mathbb{F}_{4}$-additive
code with generator matrix $G=\left(G_{X}|G_{Y}\right) .$ If $G_X G_X^{t} = I_{\alpha}$, and $G_{Y}$ generates an $\mathbb{F}_{4}$-additive self-orthogonal code, then $C$ is an
ACD code.
\end{proposition}

\begin{proof} We have $G{G}^{t} = \omega {G_X G_X^{t}}  + G_Y G_Y^{t} = \omega {G_X G_X^{t}}$, and since $G_X G_X^{t} = I_{\alpha}$, we get $G{G}^{t} = \omega I_{\alpha}$.
Therefore, according to Theorem \ref{invertible-LCD}, we conclude that $C$ is an ACD code.
\end{proof}

\begin{corollary}
\label{prop 29} Let $C$ be an $\mathbb{F}_{2}\mathbb{F}_{4}$-additive
code with generator matrix $G=\left( I_{\alpha}|G_{Y}\right) ,$ where $G_{Y}$
generates an $\mathbb{F}_{4}$-additive self-orthogonal code. Then $C$ is an
ACD code.
\end{corollary}

\begin{proof}
Suppose that $C$ is an $\mathbb{F}_{2}\mathbb{F}_{4}$-additive code with
generator matrix $G=\left( I_{\alpha}|G_{Y}\right) ,$ where $G_{Y}$ generates an $
\mathbb{F}_{4}$-additive self-orthogonal code. Then the rows of $G$ are of
the form $\left\{ V_{1},V_{2},\ldots ,V_{\alpha}\right\} ,$ where $V_{i}=\left(
a_{i}|b_{i}\right) $ and $a_{i}=\left( 0,0,\ldots ,1,0,\ldots ,0\right),$ in
which $1$ is in the ith position. Thus, 
$
\left\langle V_i, V_j \right\rangle =
\begin{cases}
\omega & \text{if } i = j, \\
0 & \text{if } i \ne j.
\end{cases}
$\\
By Theorem \ref{ACD-Three-Cases-Similar}, we get that $C$ is an
ACD code.
\end{proof}

\begin{example} 
\label{exa11: span}
Let $C$ be an $\mathbb{F}_{2}\mathbb{F}_{4}$-additive code of type $\left( 2 ,2 ;1,0,1 \right)$ generated by:
$$
G=\left(\begin{matrix}
 1 &  0 &   \bigm| & 1 & 1 \\
   0 & 1 &  \bigm| & \omega & \omega
\end{matrix}\right)
.$$
 We observe that the generator matrices $G_X$ and $G_Y$ satisfy the conditions stated in Proposition \ref {prop 29}. Consequently, $C$ is an ACD code.
\end{example}

In Theorem \ref{th23: span}, we assumed that $C_X$ is a binary LCD code and $C_Y$ is an $\F_4$-additive self-orthogonal code. Now, we change the situation. This time, $C_X$ is a binary self-orthogonal code, and $C_Y$ is an ACD code over $\F_4$. We show that, under these conditions together with one additional condition, the code $C$ is still an ACD code. This is the result of the next theorem.

\begin{theorem}
\label{th27: span}
Let $C$ be an $\F_2 \F_4$-additive code such that $C_X$ is a binary self-orthogonal code and $C_Y$ is an ACD code over $\F_4$. If the quaternary part $b$ of every nonzero codeword $u = (a , b)\in C $ is nonzero, then $C$ is an ACD code.
\end{theorem}

\begin{proof} 
The argument follows the same structure as that presented in Theorem \ref{th23: span}
\end{proof}

\begin{example}
Let $C$ be an $\mathbb{F}_{2}\mathbb{F}_{4}$-additive code of type $\left( 4 ,1 ;1,0,1 \right)$ generated by:
$$
G=\left(\begin{matrix}
1 &   1 & 1 &   1 &   \bigm| & 1  \\
  1 &   1 & 0 &  0 &  \bigm| & \omega
\end{matrix}\right)
.$$
In this case $C_X$ is a binary self-orthogonal code, $C_Y$ is an ACD code, and the quaternary part of every nonzero codeword in $ C $ is nonzero. By Theorem \ref{th27: span}, we conclude that $C$ is an ACD code.
\end{example}

We now consider the case where $C_Y$ is not an ACD code over $\F_4$. This leads to the following corollary, which is easy to verify.

\begin{corollary} 
\label{cor33: span}
Let $C$ be an $\F_2 \F_4$-additive code such that $C_X$ is a binary self-orthogonal code and $C_Y$ is not an ACD code over $\F_4$. Then $C$ is not an ACD code. 
\end{corollary}

\begin{example}
Let $C$ be an $\mathbb{F}_{2}\mathbb{F}_{4}$-additive code of type $\left( 4 ,3 ;0,1,0 \right)$ generated by:
$$
G=\left(\begin{matrix}
 1 &  1 &  1 & 1 &  \bigm| & 1 & \omega & \omega^2 \\\
   0 & 1 & 0 &  1 &  \bigm| & \omega & \omega^2 & 1
\end{matrix}\right)
.$$
 Is it clear that $C_X$ is a binary self-orthogonal code. However, since $( 1, \omega, \omega^2)\in C_Y^{\perp}$, we deduce that $C_Y$ is not an ACD code. Consequently, it follows from Corollary \ref{cor33: span} that $C$ is not an ACD code.
\end{example}

Since every self-dual code is self-orthogonal, Theorem \ref{th27: span} remains valid when $C_X$ is assumed to be self-dual code. This observation leads to the following corollary.

\begin{corollary}
\label{cor35: span}
Let $C$ be an $\F_2 \F_4$-additive code such that $C_X$ is a binary self-dual code and $C_Y$ is an ACD code over $\F_4$. If the quaternary part $b$ of every nonzero codeword $u = (a , b)\in C $ is nonzero, then $C$ is an ACD code. 
\end{corollary}

\begin{example}
Let $C$ be an $\mathbb{F}_{2}\mathbb{F}_{4}$-additive code of type $\left( 4 ,2 ;0,1,0 \right)$, with
generator matrix:
$$
G=\left(\begin{matrix}
 1 &  0 &  1 & 0 &  \bigm| & \omega & \omega^2 \\\
   0 & 1 & 0 &  1 &  \bigm| & \omega^2 & 1
\end{matrix}\right)
.$$
In this example, $C_X = C_X^{\perp}$, $C_Y$ is an ACD code over $\F_4$, and the quaternary part of every nonzero codeword in $ C $ is nonzero. So, from Corollary \ref{cor35: span}, we deduce that $C$ is an ACD code.
\end{example}

\begin{proposition}
\label{th37: span}
Let $C$ be an $\mathbb{F}_{2}\mathbb{F}_{4}$-additive
code with generator matrix $G=\left(G_{X}|G_{Y}\right),$ where $G_{X}$ generates a binary self-orthogonal code and $G_Y G_Y^{t}$ is an invertible diagonal matrix. Then $C$ is an
ACD code.
\end{proposition}

\begin{proof} 
The argument follows the same structure as that presented in Proposition \ref{th28: span}.
\end{proof}

\begin{example}
Let $C$ be an $\mathbb{F}_{2}\mathbb{F}_{4}$-additive code of type $\left( 2 ,3 ;1,1,0 \right)$, with
generator matrix:$$G=\left( 
\begin{tabular}{ll|lll}
$1$ & $1$ & $0$ & $0$ & $\omega $ \\\hline

$0$ & $0$ & $1$ & $\omega $ & $0$ \\ 
$0$ & $0$ & $\omega $ & $1 $ & $0$%
\end{tabular}%
\right).$$ Is it clear that $G_{X}$ generates a binary self-orthogonal code. On the other hand, we have $G_Y G_Y^{t}=\left( 
\begin{tabular}{lll}
$\omega^2$ & $0$ & $0$ \\ 
$0$ & $\omega$ & $0$  \\ 
$0$ & $0$ & $\omega $ %
\end{tabular}%
\right),$ with $det(G_Y G_Y^{t})=\omega$. Thus, $G_Y G_Y^{t}$ is an invertible diagonal matrix. Consequently, from Proposition \ref{th37: span}, $C$ is an ACD code.
\end{example}

\begin{proposition}
\label{prop38: span}
Let $C$ be an $\F_2 \F_4$-additive code with generator matrix $G = \left(G_X |G_Y \right)$, where $G_X$ generates a binary self-orthogonal code and $ G_Y =  \left(^{I_{\beta}}_ {\omega I_{\beta}}\right)$. Then $C$ is an $ACD$ code.
\end{proposition}

\begin{proof} We have $G{G}^{t} = \omega {G_X G_X^{t}}  + G_Y G_Y^{t} = G_Y G_Y^{t}$, and
$G_Y G_Y^{t}  = \left(\begin{matrix}
 I_{\beta}  &    \omega{I}_{\beta} \\
   \omega{I}_{\beta}  &  \omega^2{I}_{\beta}
\end{matrix}\right)$. We observe that $a_{ij} \in\{ 0, \omega\}$, for $i\not =j$ and $ a_{ii} \notin\{ 0, \omega\}$. Hence, according to Corollary \ref{coro- three-cases}, $C$ is an $ACD$ code.
\end{proof}
\begin{example} Let $C$ be an $\mathbb{F}_{2}\mathbb{F}_{4}$-additive code of type $\left( 6 ,2 ;2,0,2 \right)$, with
generator matrix:
$$
G=\left(\begin{matrix}
 1 & 1  & 1 & 1 & 1 & 1 & \bigm| & 1 & 0   \\
 1 & 1  & 0 & 0 & 0 & 0 & \bigm| & 0 & 1  \\
 0 & 0  & 1 & 1 & 0 & 0 & \bigm| & \omega & 0   \\
 1 & 1  & 1 & 1 & 0 & 0 & \bigm| & 0 & \omega
\end{matrix}\right)
.$$
It is clear that $G_X$ generates a binary self-orthogonal code. Therefore, it follows from Proposition \ref{prop38: span} that $C$ is an ACD code.
\end{example}

Note that even if the punctured codes $C_X$ and $C_Y$ are
binary LCD and quaternary ACD codes, respectively, this does not necessarily imply that the code $C$ is an ACD code, as demonstrated in the following example.

\begin{example} Let $C$ be an $\mathbb{F}_{2}\mathbb{F}_{4}$-additive code of type $\left( 2 ,4 ;1,1,0 \right)$, with
generator matrix:
$$
G=\left(\begin{matrix}
    1 & 1 & \bigm| & 0 & 0 & \omega & \omega \\
    \hline
    0 & 1 & \bigm| & 1 & \omega & 1&0  \\
   0 & 0 & \bigm| & \omega & 0 & 0& 0
\end{matrix}\right)
.$$
Since there exists a nonzero vector $(1, 1 | 0, 0, \omega, \omega ) \in C \cap C^{\perp}$, the code $C$ is not an ACD code.
Nevertheless, $C_X$ is an LCD code, and $C_Y$ is an ACD code.
\end{example}

\section{Binary LCD Codes from ACD Codes over $\mathbb{F}_2 \mathbb{F}_4$} \label{sec6:per}

In general, if $ C $ is an ACD code, its binary image is not necessarily an LCD code. In this section, we provide sufficient conditions under which the image of an ACD code over $\mathbb{F}_2 \mathbb{F}_4$ results in a binary LCD code. Moreover, we show that if the binary image of a code $C$ is an LCD code, then $C$ itself is an ACD code. 
Let us begin with the following lemma, which will be used in the proof of the next theorem.

\begin{lemma}
\label{lem 30:span}
Let $C$ be an additive code over $\mathbb{F}_4$ generated by a matrix $G = G_1 + \omega G_2$, where $G_1$ and $G_2$ are matrices with entries in $\mathbb{F}_2$. If $C$ is a self-orthogonal code with respect to the Euclidean inner product over $\mathbb{F}_4$, then the following identities hold:
$$
G_1 G_1^{t} + G_2 G_2^{t} = 0 \quad \text{and} \quad G_1 G_2^{t} + G_2 G_1^{t} + G_2 G_2^{t} = 0.
$$
\end{lemma}

\begin{proof}
Since $C$ is a self-orthogonal code, we have
$$
G G^{t} = G_1 G_1^{t} + G_2 G_2^{t} + \omega (G_1 G_2^{t} + G_2 G_1^{t} + G_2 G_2^{t}) = 0.
$$
This implies that
$
G_1 G_1^{t} + G_2 G_2^{t} = 0 \quad \text{and} \quad G_1 G_2^{t} + G_2 G_1^{t} + G_2 G_2^{t} = 0.
$
\end{proof}

\begin{theorem}
\label{the 41}
Let $C$ be an ACD code over $\F_2 \F_4$ with generator matrix $G = \left( G_X  \,\middle|\, G_Y \right)$, where $ G_X$ generates a binary LCD code and $ G_Y $ generates a self-orthogonal additive code over $\mathbb{F}_4$. If the binary part of every nonzero codeword in $C$ is nonzero, then $C' = W(C)$ is a binary LCD code.
\end{theorem}

\begin{proof}
Let $G_Y = G_1 + \omega G_2$, where $G_1$ and $G_2$ are matrices with entries in $\mathbb{F}_2$, and let $C$ be an ACD code generated by
$
G = \left( G_X \,\middle|\, G_1 + \omega G_2 \right)
.$ Applying the map $W$ to $G$, we obtain
$
G' = W(G) = \left(G_X \,\middle|\, G_1 + G_2 \,\middle|\, G_2 \right).
$
The Gram matrix is computed as:
\begin{align*}
G' {G'}^{t} = G_X {G_X}^{t}+ G_1 G_1^{t} + G_1 G_2^{t} + G_2 G_1^{t}.
\end{align*} Using Lemma~\ref{lem 30:span}, we obtain
$G' {G'}^{t} = G_X{G_X}^{t}.$
Since $ G_X$ generates a binary LCD code, and the binary part of every nonzero codeword in $C$ is nonzero, we have $\det(G_X {G_X}^{t}) \ne 0$. Therefore, by Proposition \ref{LCD codes}, $C'$ is a binary LCD code.
\end{proof}

We now present two corollaries that follow directly from Theorem \ref {the 41} and are easy to verify.

\begin{corollary}
\label{th45: span}
Let $C$ be an ACD code over $\F_2 \F_4$ with generator matrix $G=\left(G_{X}|G_{Y}\right) .$ If $G_X G_X^{t} = I_{\alpha}$ and $G_{Y}$ generates an $\mathbb{F}_{4}$-additive self-orthogonal code, then $C' = W(C)$ is a binary LCD code.
\end{corollary}

\begin{corollary}
\label{prop46:span}
Let $C$ be an ACD code over $\F_2 \F_4$ generated by $G = \left( I_\alpha \,\middle|\, G_Y \right)$, where $G_Y $ generates a self-orthogonal additive code over $\mathbb{F}_4$. Then  $C' = W(C)$ is a binary LCD code.
\end{corollary}

\begin{example}
Let $C$ be an $\mathbb{F}_{2}\mathbb{F}_{4}$-additive code of type $\left( 2 ,3 ;0,1,0 \right)$ generated by:
$$
G = \begin{pmatrix}
1 & 0 & \bigm| & \omega & \omega^2 & 1 \\
0 & 1 & \bigm| & \omega^2 & 1 & \omega
\end{pmatrix}.
$$ It is clear that $G_Y$ generates a self-orthogonal additive code over $\mathbb{F}_4$.  
Therefore, by Corollary \ref{prop46:span}, $C'$ is a binary LCD code.
\end{example}

From the proof of Theorem \ref{the 41}, we observe that if $G_X {G_X}^{t}=0$, and $det( G_1 G_1^{t} + G_1 G_2^{t} + G_2 G_1^{t})\neq0$, then $det(G' {G'}^{t})\neq 0$. This leads us to the following theorem.

\begin{theorem}
\label{th 44:span}
Let $C$ be an $\F_2 \F_4$-additive code with generator matrix $G = \left(G_X |G_1 + \omega G_2 \right)$, where $G_1$ and $G_2$ are matrices with entries in $\mathbb{F}_2$. If $G_X$ generates a binary self-orthogonal code and $det( G_1 G_1^{t} + G_1 G_2^{t} + G_2 G_1^{t})\neq0,$ then $C'= W (C)$ is a binary LCD code.
\end{theorem}

\begin{proof} The proof follows the same reasoning as the proof of Theorem \ref{the 41}.
\end{proof}

\begin{example}
Let $C$ be an $\mathbb{F}_{2}\mathbb{F}_{4}$-additive code of type $\left( 2 ,2 ;1,0,1 \right)$ generated by:
$$
G = \begin{pmatrix}
1 & 1 & \bigm| & \omega & \omega^2  \\
0 & 0 & \bigm| & 0 & 1  
\end{pmatrix},
$$ with $
G_1 = \begin{pmatrix}
0 & 1 \\
0 & 1   
\end{pmatrix}$, 
$G_2 = \begin{pmatrix}
1 & 1  \\
0 & 0   
\end{pmatrix}
$, and  $ G_1 G_1^{t} + G_1 G_2^{t} + G_2 G_1^{t} = \begin{pmatrix}
1 & 0  \\
0 & 1   
\end{pmatrix}$. Since $det( G_1 G_1^{t} + G_1 G_2^{t} + G_2 G_1^{t}) = 1,$  and $G_X$ generates a binary self-orthogonal code, it follows, by Theorem
\ref {th 44:span}, that $C'$ is a binary LCD code.
\end{example}

The following corollary is an immediate consequence of Theorem \ref{th 44:span}.

\begin{corollary}
Let $C$ be an $\F_2 \F_4$-additive code with generator matrix $G = \left(G_X |G_1 + \omega G_2 \right)$, where $G_1$ and $G_2$ are matrices with entries in $\mathbb{F}_2$. If $G_X$ generates a binary self-orthogonal code and $det( G_1 G_1^{t} + G_1 G_2^{t} + G_2 G_1^{t}) = 0,$ then $C'= W (C)$ is not a binary LCD code.
\end{corollary}

\begin{proposition}
\label{th 30:span}
Let $C$ be an ACD code over $\F_2 \F_4$, with generator matrix $G = \left(G_X |G_Y \right)$, where $G_X$ generates a binary self-orthogonal code and $ G_Y = \left(^{I_{\beta}}_ {\omega I_{\beta}}\right)$. Then $C'= W (C)$ is a binary LCD code.
\end{proposition}

\begin{proof} Let $C$ be an $\F_2 \F_4$-additive code with generator matrix $G = \left(G_X |G_Y \right)$, where $ G_Y =  \left(^{I_{\beta}}_ {\omega I_{\beta}}\right)$, and let $G' = W (G)$ be the generator matrix of $C'= W (C)$. We have $G'{G'}^{t} = G_X G_X^ {t} +  \begin{pmatrix}
 I_{\beta} &  I_{\beta} \\  I_{\beta}&  0_{\beta}
\end{pmatrix}$, where $0_{\beta}$ is the zero matrix of size $\beta$. Since $G_X$ generates a binary self-orthogonal code, then $G'{G'}^{t}= \begin{pmatrix}
 I_{\beta} &  I_{\beta} \\  I_{\beta}&  0_{\beta}
\end{pmatrix} $. \\$ det ( G'{G'}^{t}) = det(I_{\beta}) = 1 \not= 0$. It follows from Proposition \ref{LCD codes} that $C' = W (C)$ is a binary LCD code.
\end{proof}

\begin{example} Let $C$ be an $\mathbb{F}_{2}\mathbb{F}_{4}$-additive code of type $\left( 6 ,2 ;2,0,2 \right)$, with
generator matrix:
$$
G=\left(\begin{matrix}
 1 & 1  & 1 & 1 & 1 & 1 & \bigm| & 1 & 0   \\
 1 & 1  & 0 & 0 & 0 & 0 & \bigm| & 0 & 1  \\
 0 & 0  & 1 & 1 & 0 & 0 & \bigm| & \omega & 0   \\
 1 & 1  & 1 & 1 & 0 & 0 & \bigm| & 0 & \omega
\end{matrix}\right)
.$$
Since $G_X$ generates a binary self-orthogonal code, it follows, by Proposition \ref{th 30:span}, that $C'$ is a binary LCD code.
\end{example}

In general, verifying whether a code is a linear complementary dual (LCD) code is easier than checking whether it is an additive complementary dual (ACD) code. Therefore, the following important theorem shows that if the image of a code $C$ is an LCD code, then $C$ itself is an ACD code.

\begin{theorem}
\label{th image not ACD}
Let $C$ be an $\F_2\F_4$-additive code of length $n=\alpha+\beta$, and let $C'= W (C)$ be its binary image code of length $n'=\alpha+2\beta$. If $C$ is not an additive complementary dual (ACD) code, then its image $C'= W (C)$ is not a linear complementary dual (LCD)code.
\end{theorem}

\begin{proof} Assume that $C$ is not an additive complementary dual code. Then there exists a nonzero vector $ u \in C \cap C^{\perp}$. Since $ u \in C $, and $W$ is a linear bijection, it follows that $ W(u)\neq0 \in W(C) = C'$. \\ Also, since $ u \in C^{\perp}$, we have $ W(u) \in W(C^\perp)$.
From lemma \ref{W(Dual-Code)}, we obtain that $ W(u) \in (W(C))^\perp = (C')^\perp$, which means $ W(u) \in C' \cap (C')^\perp$. Thus, $ W(u)$ is a nonzero codeword in  $ C'\cap {C'}^\perp$. Therefore, $C'= W (C)$ is not an LCD code.
\end{proof}\\
The converse of Theorem \ref{th image not ACD} does not hold in general, as illustrated by the following example.

\begin{example} Let $C$ be an $\mathbb{F}_{2}\mathbb{F}_{4}$-additive code of type $\left( 2 ,3 ;1,1,0 \right)$, with
generator matrix:
$$
G=\left(\begin{matrix}
    1 & 1 & \bigm| & 0 & 0 & \omega \\\hline
    0 & 1 & \bigm| & 1 & \omega & 1 \\
   0 & 0 & \bigm| & \omega & 0 & 0
\end{matrix}\right)
,$$
and let $ C' = W(C)$ be a binary linear code with generator matrix $G' = W(G)$, defined as follows:
$$ G' = W(G) = \begin{pmatrix}
1 & 1  & 0 & 0 & 1 & 0  & 0 & 1 \\
0 & 1  & 1 & 1 & 1 & 0 &  1 & 0 \\
0 & 0  & 1 & 0 & 0 & 1 &  0 & 0\\
\end{pmatrix}.$$
We compute the Gram matrix: $$G'{G'}^{t}  = \begin{pmatrix}
0  &  0& 0\\ 0 & 1& 1 \\ 0 & 1& 0
\end{pmatrix}.$$
Observe that $ det ( G'{G'}^{t}) = 0 $. Hence, by Proposition \ref{LCD codes}, $C'$ is not an LCD code. However, since $C \cap C^\perp = \mathbf\{0\}$, we conclude that $C$ is an ACD code. This shows that the converse of Theorem \ref{th image not ACD} does not hold in general.
\end{example}

Since the image of an ACD code is not necessarily an LCD code, the next result provides a full characterization of when this equivalence holds. Based on all the previous results, we are now able to state the precise conditions under which $C$ is ACD if and only if its image is LCD.

\begin{proposition} Let $C$ be an $\F_2\F_4$-additive code with generator matrix $G = \left(G_X |G_Y \right)$, and let $C'= W (C)$ be its binary image. If one of the following conditions holds: 
\begin{enumerate}
    \item $ G_X $ generates a binary LCD code, $ G_Y $ generates an $\mathbb{F}_{4}$-additive self-orthogonal code, and the binary part of every nonzero codeword in $C$ is nonzero.
   \item $G_X G_X^{t} = I_{\alpha}$, and $G_{Y}$ generates an $\mathbb{F}_{4}$-additive self-orthogonal code.
    \item $ G_X = I_{\alpha}$, and $ G_Y $ generates an $\mathbb{F}_{4}$-additive self-orthogonal code . 
    \item $ G_X $ generates a binary self-orthogonal code, and $ G_Y = \left(^{I_{\beta}}_ {\omega I_{\beta}}\right)$.
\end{enumerate}
Then, $C$ is an ACD code if and only if $C'= W (C)$ is an LCD code.
\end{proposition}

\begin{remark}  
Since a self-dual code is, by definition, always self-orthogonal, the results we gave for the punctured code $ C_X$, under the assumption of self-orthogonality, also apply to self-dual codes. Therefore, all the theorems, propositions, and corollaries we presented remain valid when working with self-dual codes in that context.
\end{remark}

We now present an example of an ACD code such that its image is a good binary LCD code.

\begin{example} Let $C$ be an $\mathbb{F}_{2}\mathbb{F}_{4}$-additive code of type $\left( 4 ,6 ;2,2,0 \right)$, with
generator matrix:
$$
G=\left[
\begin{array}{cccc|cccccccccccc}
1&0&1&1&0&0&0&0&\omega&0\\
0&1 &0&1&0&0 &0&0 &\omega&\omega\\ 
\hline
0&0 &1&0&1&0 &\omega&0 &1&1  \\
0&0 &0&1 &0&1 &0&\omega &1&0 \\ 
0&0 &0&0&\omega&0 &\omega&\omega &1&0  \\
0&0 &0&0 &0&\omega &\omega&0 &0&1 
\end{array}
\right]
.$$   
Let $ C' = W(C)$ denote the binary image of $C$, where $ C'$ is generated by the matrix $G' = W(G)$ defined as follows:
$$
G'=\left[
\begin{array}{cccccccccccccccc}
1&0&1&1&0&0&0&0&1&0&0&0&0&0&1&0\\
0&1&0&1&0&0&0&0&1&1&0&0&0&0&1&1\\ 
0&0&1&0&1&0&1&0&1&1&0&0&1&0&0&0\\
0&0&0&1&0&1&0&1&1&0&0&0&0&1&0&0\\ 
0&0&0&0&1&0&1&1&1&0&1&0&1&1&0&0\\
0&0&0&0&0&1&1&0&0&1&0&1&1&0&0&0 
\end{array}
\right]
.$$   
The best known code with parameters $[16, 6]$ has a minimum distance of $6$. Our code has parameters $[16, 6, 5]$, so it is not optimal. But since the minimum distance is only one less, we can still say that it is a good code.\\
We compute the Gram matrix:  
 $$G'{G'}^{t}  =  \begin{pmatrix}
1  &  1& 0&  0& 1& 0\\ 1  &  0& 0&  0& 1& 1 \\ 0  &  0& 0&  1& 0& 1\\0  &  0& 1&  1& 1& 1\\1  &  1& 0& 1 & 1& 0\\0  &  1& 1&  1& 0& 1
\end{pmatrix},$$
and $det ( G'{G'}^{t}) = 1 $. Therefore, by Proposition \ref{LCD codes}, $C'$ is an LCD code. Furthermore, by applying Theorem \ref{th image not ACD}, we deduce that $C$ is an ACD code.
\end{example}

We now present another example of an additive complementary dual (ACD) code over $\F_2\F_4$, 
whose binary image is an optimal linear complementary dual (LCD) code. 

\begin{example} Let $C$ be an $\mathbb{F}_{2}\mathbb{F}_{4}$-additive code of type $\left( 3 ,2 ;1,0,1 \right)$, defined by the following generator matrix:
$$
G=\left(\begin{matrix}
    1 & 1 & 1 &\bigm|  & \omega & \omega^{2}\\
    0 & 0 & 0 &\bigm| & \omega &  \omega
\end{matrix}\right)
,$$ and let $ C' = W(C)$ be the binary image of $C$ generated by $G' = W(G)$, defined as follows:
$$ G' = W(G) = \begin{pmatrix}
1 & 1  & 1 & 1 & 0 & 1  & 1 \\
0 & 0  & 0 & 1 & 1 & 1 & 1  
\end{pmatrix}.$$
We have $det ( G'{G'}^{t}) = 1 $. Therefore, by Proposition \ref{LCD codes}, $C'$ is an LCD code.
Since the minimum distance is $4$, which is the best known for binary codes of length $7$ and dimension $2$, the code $C'$ can be considered an optimal LCD code.
Applying Theorem \ref{th image not ACD} leads us to conclude that $C$ is an ACD code.
\end{example}

As we did previously, we now give the final example, which shows that the image of an ACD code can be an optimal binary code that is not LCD. 

\begin{example} Let $C$ be an $\mathbb{F}_{2}\mathbb{F}_{4}$-additive code of type $\left( 3 ,2 ;0,1,0 \right)$ generated by:
$$
G=\left(\begin{matrix}
    1 & 1 & 1 &\bigm|  & 1 & \omega^{2}\\
    1 & 1 & 1 &\bigm| & \omega^{2} &  1 
\end{matrix}\right)
,$$ and let $ C' = W(C)$ be the binary image of $C$ generated by $G' = W(G)$, defined as follows:
$$ G' = W(G) = \begin{pmatrix}
1 & 1  & 1 & 1 & 0 & 0  & 1 \\
1 & 1  & 1 & 0 & 1 & 1 & 0  
\end{pmatrix}.$$
We have $det ( G'{G'}^{t}) = 0 $. Therefore, by Proposition \ref{LCD codes}, $C'$ is not an LCD code. 
Because the minimum distance equals the best known for a $[7,2]$ binary code, we consider $C'$ to be an optimal code.
Since $C$ satisfies the second condition of Theorem \ref{ACD-Three-Cases}, it follows that it is an ACD code.
\end{example}

\section{Conclusion and future work}\label{sec7:per}

In this paper, we investigated the structure and properties of additive
complementary dual (ACD) codes over the mixed alphabets $\mathbb{F}_{2}\mathbb{F}%
_{4}$ under a certain inner product defined over the ring $\mathbb{F}_{2}%
\mathbb{F}_{4}.$ We have provided results and established sufficient
conditions under which such codes are (ACD) codes. We have also shown that
ACD codes over $\mathbb{F}_{2}\mathbb{F}_{4}$ can be applied to construct
binary linear complementary dual codes as images of ACD codes under the linear map $W$. In the final part of this work, we showed that if the binary image of an $\mathbb{F}_2\mathbb{F}_4$-additive code is an LCD code, then the original code must also be an ACD code. We have also provided an example of an ACD code whose binary image is a distance-optimal binary LCD code. For future work, it will
be interesting to study ACD codes over $\mathbb{F}_{2}\mathbb{F}_{4}$ with
different definitions of inner products and study the conditions needed to
construct ACD codes.

\end{document}